\documentclass[10pt,a4paper]{amsart}

\usepackage{fullpage}
\usepackage{amsmath}
\usepackage{amsfonts}
\usepackage{color}
\usepackage{url}

\usepackage{amsthm}
\usepackage[all]{xy}

\title{Codes from Jacobian surfaces}
\author{Safia Haloui}
\address{Department of Mathematics, Technical University of Denmark, Lyngby, Denmark}
\email{safia.haloui@gmail.com}
\date{\today}

\newtheorem{theo}{Theorem}%[section]
\newtheorem{prop}[theo]{Proposition}

\theoremstyle{definition}

\theoremstyle{remark}
\newtheorem{rem}[theo]{Remark}
\begin{document}

\begin{abstract}
This paper is concerned with some Algebraic Geometry codes on Jacobians of genus $2$ curves. We derive a lower
bound for the minimum distance of these codes from an upper ``Weil type''  bound for
 the number of rational points on irreducible (possibly singular or non-absolutely 
irreducible) curves lying on an abelian surface over a finite field.
\end{abstract}

\maketitle
%%%%%%%%%%%%%%%%%%%%%%%%%%%%%%%%%%%%%%%%%%%%%%%%%%%%%%%
%%%%%%%%%%%%%%%%%%%%%%%%%%%%%%%%%%%%%%%%%%%%%%%%%%%%%%%
%About the Weil polynomials : generalisation en dim sup ?
%%%%%%%%%%%%%%%%%%%%%%%%%%%%%%%%%%%%%%%%%%%%%%%%%%%%%%%
\section{Introduction}

Algebraic Geometry codes were introduced by V.D. Goppa in the beginning of the 80's. While A.G. codes arising from curves have been extensively studied and are now fairly well-understood, only few things are known about the higher dimension case and it would be interesting to get new examples. 

\bigskip

This paper is concerned with \emph{evaluation codes} (in the sense of \cite{lit}) from Jacobian surfaces, defined by a very ample divisor numerically equivalent to a positive multiple of a divisor corresponding to the principal polarization. Given a smooth, projective, absolutely irreducible genus $2$ curve $C/\mathbb{F}_q$ and a positive integer $r$, the linear code $C(J_C,G)$ defined in Section \ref{DefEvCodes}, where $J_C$ is the Jacobian of $C$ and $G$ is a very ample divisor which is numerically equivalent to $rC$ has length
$$n=\# J_C(\mathbb{F}_q)=\frac{1}{2}\left(\# C(\mathbb{F}_{q^2})+\# C(\mathbb{F}_q)^2\right)-q.$$
We prove that if $J_C$ is simple over ${\mathbb F}_q$, then the minimum distance of $C(J_C,G)$ satisfies
$$d\geq \# J_C(\mathbb{F}_q)- \max \left\{\# C(\mathbb{F}_q)+(r^2-1)[2\sqrt{q}],\quad r\# C(\mathbb{F}_q)\right\}.$$
In particular, if
$$r\leq \frac{\# C(\mathbb{F}_q)}{[2\sqrt{q}]}-1,$$
then we have
\begin{eqnarray}\label{BorneMiniDist}
d\geq \# J_C(\mathbb{F}_q)-r\# C(\mathbb{F}_q).
\end{eqnarray}
Moreover, if this lower bound for the minimum distance is positive, then $C(J_C,G)$ has dimension
$$k=r^2.$$

The lower bound (\ref{BorneMiniDist}) goes in the direction of the intuitive idea that minimum weight codewords should come from effective divisors in the linear system of $G$ with the greatest possible number of components (see the proof of Theorem \ref{MinDistJac}). As an example, if $G=rC$ and there exists $P_1,\dots ,P_r\in J_C(\mathbb{F}_q)$ such that $P_1+\dots +P_r=0$, then, by the Theorem of the Square (see \cite{mil2},\cite{mum}), G is linearly equivalent to the sum $D=(C+P_1)+\dots +(C+P_r)$ of translates of $C$ and if the $(C+P_i)$'s do not intersect at rational points, then (\ref{BorneMiniDist}) is attained by the codewords arising from $D$. Notice that all the prime components of $D$ have arithmetic genus $2$, which is the smallest possible for an irreducible curve over $\mathbb{F}_q$ on $J_C$ in the case where $J_C$ is simple.

\bigskip

This paper is organized as follows: Section \ref{CourbeOnAbSurf} is devoted to the study of the number of rational points on projective irreducible (not necessary absolutely irreducible) curves lying on an abelian surface over a finite field. In particular, we give a ``Weil type bound'' for these curves, depending on the trace of the abelian variety.
In Section \ref{CodesJacobian}, after giving some basic definitions about A.G. codes on surfaces, we use the results from the former section to derive the lower bound on the minimum distance of $C(J_C,G)$ mentioned above.
%%%%%%%%%%%%%%%%%%%%%%%%%%%%%%%%%%%%%%%%%%%%%%%%%%%%%%%
%%%%%%%%%%%%%%%%%%%%%%%%%%%%%%%%%%%%%%%%%%%%%%%%%%%%%%%
%%%%%%%%%%%%%%%%%%%%%%%%%%%%%%%%%%%%%%%%%%%%%%%%%%%%%%%
%%%%%%%%%%%%%%%%%%%%%%%%%%%%%%%%%%%%%%%%%%%%%%%%%%%%%%%

\section{Curves on abelian surfaces over finite fields}\label{CourbeOnAbSurf}

\subsection{About the Weil polynomials}\label{AboutWP}

Let $A$ be an abelian variety of dimension $g$ defined over the finite field  ${\mathbb F}_q$ with $q = p^{a}$ elements, where $p$ is a prime number. The \emph{Weil polynomial} $f_A$ of $A$ is the characteristic polynomial of its Frobenius endomorphism acting on the $\mathbb{Q}_\ell$-vector space
$T_\ell(A)\otimes_{\mathbb{Z_\ell}}\mathbb{Q}_\ell$,
where $\ell$ is a prime number distinct from $p$ and
$$T_{\ell}(A)=\lim_{\leftarrow}A[\ell^{n}](\overline{\mathbb F_q})$$
is the Tate module of $A$; it gives an explicit characterization of the isogeny class of $A$. The Weil polynomial has the following property which is convenient for our purpose: if $A$ is the Jacobian of some smooth, projective, absolutely irreducible curve over ${\mathbb F}_q$, then the number of rational points on this curve can be easily derived from $f_A$. We will come back to that in the next section. Notice also that the number of rational points on $A$ is equal to $f_A(1)$. More information on Weil polynomials can be found in J.S. Milne's surveys \cite{mil2} and \cite{mil3}.

\bigskip

Now, we suppose that $A/{\mathbb F}_q$ is an abelian surface. Let $D$ be a projective, absolutely irreducible algebraic curve defined over $\mathbb{F}_q$, lying on $A$ (i.e. $D$ is a closed algebraic subvariety of $A$), and let $\widetilde{D}$ be the normalization of $D$. The composite of the normalization map $n :\widetilde{D}\rightarrow D $ with the inclusion $i:D\hookrightarrow A$  gives rise to a map $\widetilde{D}\rightarrow A $ which is birational onto its image. Since $i\circ n$ is not constant, the genus of $\widetilde{D}$ is nonzero (\cite{mil2}, Cor. 3.8).

If $\widetilde{D}$ has genus $1$, then the map $i\circ n$ is actually an embedding, and therefore $D$ is nonsingular. Indeed, in this case $\widetilde{D}$ has a structure of elliptic curve (any smooth, projective, absolutely irreducible genus $1$ curve over a finite field has a rational point), and the map $i\circ n$ is the composite of an homomorphism with a translation (\cite{mil2}, Cor. 2.2). Therefore, its image $D$ is a translate of an abelian subvariety of $A$, which must be nonsingular and the map $i\circ n$ is actually an embedding, since it is birational onto its image.

In this setting, a classical reasoning on the Tate modules gives us the possible Weil polynomials for $D$:

\begin{prop}\label{PolyCarEll}
Let $D/\mathbb{F}_q$ be a projective absolutely irreducible curve of geometric genus $1$ lying on an abelian surface $A/\mathbb{F}_q$. Then $D$ has a structure of elliptic curve and  $f_{D}$ divides $f_A$.
\end{prop}
\begin{proof}
Taking in account the discussion above, it remains to prove that $f_D$ divides $f_A$, and possibly modifying $D$ with a translation by a rational point (which does not change the Weil polynomial), we can assume that $D$ is an abelian subvariety of $A$. In this setting, there exists an abelian subvariety $B$ of $A$ such that the addition law of $A$ induces an isogeny
$$m: D\times B\rightarrow A$$
(see \cite{mil2}, Prop. 12.1). Now the map induced by  $m$ on the Tate modules gives rise to a ${\mathbb Q}_{\ell}$-vector space isomorphism
$$T_{\ell}(D)\otimes_{{\mathbb Z}_{\ell}}{\mathbb Q}_{\ell}\times T_{\ell}(B)\otimes_{{\mathbb Z}_{\ell}}{\mathbb Q}_{\ell}=T_{\ell}(D\times B)\otimes_{{\mathbb Z}_{\ell}}{\mathbb Q}_{\ell}{\tilde{\rightarrow}} T_{\ell}(A)\otimes_{{\mathbb Z}_{\ell}}{\mathbb Q}_{\ell}$$
(where $\ell$ is a prime number coprime to $q$), which commutes with the action of the Frobenius, since it arises from a rational map. This gives us the required factorization of the characteristic polynomial $f_A$ and completes the proof.
\end{proof}

In the case where $\widetilde{D}$ has genus greater than or equal to $2$, the following proposition gives us some information on the Weil polynomial of $J_{\widetilde{D}}$:

\begin{prop}\label{PolyCarGenreSup}
Let $D/\mathbb{F}_q$ be a projective absolutely irreducible curve lying on an abelian surface $A/\mathbb{F}_q$ and let $\widetilde{D}$ be its normalization. Suppose that $\widetilde{D}$ has a rational point and has genus greater than or equal to $2$, then $f_A$ divides $f_{J_{\widetilde{D}}}$.
\end{prop}
\begin{proof}
A rational point $P$ on $\widetilde{D}$ gives rise to an embedding $h^P:\widetilde{D}\hookrightarrow J_{\widetilde{D}}$ such that any morphism from $\widetilde{D}$ to an abelian variety sending $P$ to zero factors through $h^P$, followed by an homomorphism (see \cite{mil3}, Prop. 6.1). Applying this to the composite of the normalization map $n:\widetilde{D}\rightarrow D $ with the inclusion $i:D\hookrightarrow A$ (and possibly modifying $A$ by a translation, so that $P$ is mapped to zero), we get an homomorphism $\alpha : J_{\widetilde{D}}\rightarrow A$ such that the following diagram is commutative:
 $$ \xymatrix{
    \widetilde{D} \ar@{^{(}->}[r]^{h^P} \ar@{->>}[d]_n & J_{\widetilde{D}} \ar[d]^\alpha \\
    D \ar@{^{(}->}[r]_i & A
  }$$
The homomorphism $\alpha$ is surjective. Indeed, the image of $\alpha$ is an abelian subvariety of $A$, so it must be either the whole $A$, or an elliptic curve, or zero. The two last possibilities should be excluded, as the image of $\alpha$ contains $D$, which is birational to a nonsingular curve of genus greater than or equal to $2$.

Now let $B_0$ be the connected component of the kernel of $\alpha$ containing zero; this is an abelian subvariety of $J_{\widetilde{D}}$ (the image of $B_0$ by the Frobenius is a connected component of the kernel of $\alpha$, and contains zero because it is a rational point, thus $B_0$ is defined over $\mathbb{F}_q$). Then there exists an abelian subvariety $B_1$ of $J_{\widetilde{D}}$ such that the addition law of $J_{\widetilde{D}}$ induces an isogeny
$$m: B_0\times B_1\rightarrow J_{\widetilde{D}}$$
(\cite{mil2}, Prop. 12.1). Since $m$ is a finite map, we have $\dim B_0 + \dim B_1 =\dim J_{\widetilde{D}}$, and by the Theorem on the dimension of fibres (\cite{sha}, I.6.3), $\dim B_0 \geq \dim J_{\widetilde{D}}-\dim A=\dim J_{\widetilde{D}}-2$. Therefore, $\dim B_1 \leq 2$.

Now consider the map
$$\alpha\circ m\vert_{B_1}:B_1\rightarrow A.$$
Since $\alpha$ is surjective and $B_0$ is mapped to zero, it is surjective. We deduce that $\dim B_1 = 2$ and $\alpha\circ m\vert_{B_1}$ is an isogeny. 

Finally, as in the proof of Proposition \ref{PolyCarEll}, the maps induced by the isogenies $\alpha\circ m\vert_{B_1}$ and $m$ on the Tate modules, give rise to two ${\mathbb Q}_{\ell}$-vector space isomorphisms
$$T_{\ell}(B_1)\otimes_{{\mathbb Z}_{\ell}}{\mathbb Q}_{\ell}{\tilde{\rightarrow}} T_{\ell}(A)\otimes_{{\mathbb Z}_{\ell}}{\mathbb Q}_{\ell}$$
and
$$T_{\ell}(B_0)\otimes_{{\mathbb Z}_{\ell}}{\mathbb Q}_{\ell}\times T_{\ell}(B_1)\otimes_{{\mathbb Z}_{\ell}}{\mathbb Q}_{\ell}{\tilde{\rightarrow}} T_{\ell}(J_{\widetilde{D}})\otimes_{{\mathbb Z}_{\ell}}{\mathbb Q}_{\ell}$$
 which commute with the action of the Frobenius. This gives us the required factorization of the characteristic polynomial $f_{J_{\widetilde{D}}}$.
\end{proof}

%%%%%%%%%%%%%%%%%%%%%%%%%%%%%%%%%%%%%%%%%%%%%%%%%%%%%%%
%%%%%%%%%%%%%%%%%%%%%%%%%%%%%%%%%%%%%%%%%%%%%%%%%%%%%%%

\subsection{A Weil type bound for curves on an abelian surface}
As in the beginning of Section \ref{AboutWP}, let $A/{\mathbb F}_q$ be an abelian variety of dimension $g$. The Weil polynomial $f_A(t)$ of $A$ has degree $2g$, has integer coefficients and the set of its roots (with multiplicity) consists of couples of conjugated complex numbers $\omega_{1}, \overline{\omega}_{1}, \dots, \omega_{g}, \overline{\omega}_{g}$ of modulus $\sqrt{q}$. For $1 \leq i \leq g$, we set $x_i = -(\omega_{i} + \overline{\omega}_{i})$ and define
$$\tau(A) = - \sum_{i = 1}^{g} (\omega_{i} + \overline{\omega}_{i}) = \sum_{i=1}^{g} x_{i}.$$
We say that $A$ \emph{has trace} $-\tau (A)$.

\bigskip

If $A=J_D$ is the Jacobian of a smooth, projective, absolutely irreducible curve $D/{\mathbb F}_q$, then $f_A$ is the reciprocal polynomial of the numerator of the zeta function of $D$. This implies that the number of points on $D$ over ${\mathbb F}_{q^k}$ is 
$$\# D(\mathbb{F}_{q^k})  =  q^k +1- \sum_{i = 1}^{g} (\omega_{i}^k + \overline{\omega}_{i}^k).$$
In particular, the number of rational points on $D$ is
$$\# D(\mathbb{F}_q)  =  q+1+\tau (A).$$
Since $\vert{x_i}\vert \leq 2 \sqrt{q}$, $i=1,\dots ,g$, we have the famous \emph{Weil bound} 
$$\# D(\mathbb{F}_q) \leq q+1+2g \sqrt{q}.$$
It is actually possible to substitute  $2\sqrt{q}$ with its integer part $[2\sqrt{q}]$ in the Weil bound. Indeed, J.-P. Serre \cite{serre0} proved that for any algebraic integers $x_1,\dots , x_g\in [-2\sqrt{q} ;2\sqrt{q}]$ such that the coordinates of the $g$-tuple $( x_1,\dots ,x_{g})$ are permuted by the action of $\mbox{Gal}(\overline{\mathbb{Q}}/\mathbb{Q})$, we have
\begin{eqnarray}\label{BorneSerreWeil}
\sum_{i=1}^{g} x_{i} & \leq & g[2\sqrt{q}].
\end{eqnarray}

\bigskip

Later, Y. Aubry and M. Perret \cite{AubPer} generalized the bounds mentioned above to (possibly singular) projective, absolutely irreducible curves over finite fields. More precisely, they proved that if $D/\mathbb{F}_q$ is a projective, absolutely irreducible curve of arithmetic genus $\pi$ and geometric genus $g$ and $\widetilde{D}$ is the normalization of $D$, then we have
\begin{eqnarray}\label{AubPerSing}
\vert\# \widetilde{D}(\mathbb{F}_q)-\# D(\mathbb{F}_q)\vert\leq \pi -g.
\end{eqnarray}
Then, using the arguments mentioned in the discussion above, they deduced that we have
$$\# D(\mathbb{F}_q)\leq q+1+\pi [2\sqrt{q}].$$

\bigskip

Taking in account the results discussed in Section \ref{AboutWP} and in the beginning of this section, it is easy to derive a ``Weil type bound'' for projective, absolutely irreducible curves lying on an abelian surface, depending on the trace of the abelian surface. However, the hypothesis that the curve is absolutely irreducible is too strong for applications to coding theory; we need a result which holds for irreducible curves. 

In order to overcome this difficulty, we use some intersection theory. Let $A/{\mathbb F}_q$ be an abelian surface. Since $A$ is an algebraic group, its canonical divisor is zero (see \cite{sha}, III.6.3)  and therefore,  the Adjunction Formula (\cite{hart} Prop. V.1.5 and Exercise V.1.3) gives us that for any projective curve $D/\overline{\mathbb{F}_q}$ lying on $A$ of arithmetic genus $\pi$ the self-intersection of $D$ is
\begin{eqnarray}\label{Adjunction}
D^2=2\pi -2.
\end{eqnarray}
\begin{rem} \label{RqIntPos} If $D$ is absolutely irreducible, then $\pi$ is nonzero, since the genus of the normalization of $D$ is, and therefore, the right hand side of (\ref{Adjunction}) is non-negative. Therefore, the intersection number of any two effective divisors is always non-negative.
\end{rem}

\begin{theo}\label{NbRatPts} Let $A/{\mathbb F}_q$ be an abelian surface. If $\tau (A)\geq -q$ (for instance, this condition is always satisfied when $q\geq 16$), then for any projective irreducible curve $D/\mathbb{F}_q$ of arithmetic genus $\pi$ lying on $A$, we have
$$\# D(\mathbb{F}_q)\leq q+1+\tau (A)+\vert \pi -2\vert [2\sqrt{q}].$$
In particular, if $A$ is the Jacobian of a smooth, projective, absolutely irreducible genus $2$ curve $C/{\mathbb F}_q$ with a rational point, then we have
$$\# D(\mathbb{F}_q)\leq \# C(\mathbb{F}_q)+\vert \pi -2\vert [2\sqrt{q}].$$
\end{theo}
\begin{proof}
First, we prove the result when $D$ is absolutely irreducible. Suppose that this is the case, and write $\widetilde{D}$ for the normalization of $D$, $g$ for the genus of $\widetilde{D}$, and $x_1,\dots , x_g$ for the sums of two complex conjugated roots of the Weil polynomial $f_{J_{\widetilde{D}}}$ of the Jacobian $J_{\widetilde{D}}$, as defined at the beginning of this section.

If $g=1$, then Proposition \ref{PolyCarEll} asserts that $\pi =1$ as well, and $f_{D}=f_{J_{\widetilde{D}}}$ divides $f_A$. Then $\tau (A)=\tau (D)+x_2$ for some integer $x_2\geq -[2\sqrt q]$ and therefore, $\# D(\mathbb{F}_q)=q+1+\tau (A)-x_2\leq q+1+\tau (A)+[2\sqrt q]$.

Now suppose that $g\geq 2$. If $\widetilde{D}$ has no rational point, then the proof of the proposition is straightforward from (\ref{AubPerSing}), since $\pi -g\leq \pi -2$. Therefore, we can assume that $\widetilde{D}$ has a rational point. By Proposition \ref{PolyCarGenreSup}, the Weil polynomial $f_A$ divides $f_{J_{\widetilde{D}}}$ so we can renumber the $x_i$'s so that $x_1, x_2$ correspond to $f_A$. With this notations, we have
\begin{eqnarray}\label{nbD}
\# \widetilde{D}(\mathbb{F}_q)=q+1+\tau (A)+\sum_{i=3}^{g}x_i.
\end{eqnarray}
The coordinates of $( x_3,\dots ,x_{g})$ are permuted by the action of  $\mbox{Gal}(\overline{\mathbb{Q}}/\mathbb{Q})$, since the coordinates of $(x_1,x_{2})$ and $(x_1,\dots ,x_{g})$ are. Therefore, by (\ref{BorneSerreWeil}), the sum in the right hand side of (\ref{nbD}) is less than or equal to $(g-2)[2\sqrt q]$. Then, the Aubry-Perret Inequality (\ref{AubPerSing}) gives us
\begin{eqnarray*}
\# D(\mathbb{F}_q) & \leq & \# \widetilde{D}(\mathbb{F}_q)+\pi -g\\
 & \leq & q+1+\tau (A)+(g-2)[2\sqrt q]+(\pi -g)[2\sqrt q]\\
 & = & q+1+\tau (A)+(\pi -2)[2\sqrt q].
\end{eqnarray*}

\bigskip

It remains to prove the result when $D$ is reducible over $\overline{\mathbb{F}_q}$. Suppose that this is the case and write $D_1,\dots ,D_k$ for its absolutely irreducible components. Then the absolute Galois group $\mbox{Gal}(\overline{\mathbb{F}_q}/\mathbb{F}_q)$ acts transitively on the set of the $D_i$'s (because an orbit of this action is an irreducible component of $D$ over $\mathbb{F}_q$). A rational point on $D$ is stable by the action of $\mbox{Gal}(\overline{\mathbb{F}_q}/\mathbb{F}_q)$, so it must lie on the intersection of all the $D_i$'s. It follows that
\begin{eqnarray*}
\# D(\mathbb{F}_q) & \leq & \#\left(\bigcap_{i=1}^kD_i\right)(\overline{\mathbb{F}_q})\\
 & \leq & \#\left( D_1\cap D_2\right)(\overline{\mathbb{F}_q})\\
 & \leq & D_1.D_2=\frac{1}{2}(D_1.D_2+D_2.D_1)\\
 & \leq & \frac{1}{2}\sum_{i=1}^k\sum_{j=1}^kD_i.D_j\\
 & = & \frac{1}{2}D^2\\
 & = & \pi -1=1+\pi -2\\
 & \leq & q+1+\tau (A)+(\pi -2)[2\sqrt{q}],
\end{eqnarray*}
where the 4th row is deduced from Remark \ref{RqIntPos}, the 6th row comes from (\ref{Adjunction}) and the last row comes from the assumption that $\tau (A)\geq -q$. This concludes the proof.
\end{proof}

%\begin{rem}
%According to the proof above, if we make the stronger hypothesis that $D$ is absolutely irreducible, the we can get read of the condition $\tau (A)\geq -q$. Moreover, it is easy to see that under this additional assumption, we have the ``symmetric lower bound''
%$$\# D(\mathbb{F}_q)\geq q+1+\tau (A)-\vert \pi -2\vert [2\sqrt{q}].$$
%In particular, these two bounds are equalities for $\pi =2$. This lower bound does not necessarily hold for non-absolutely irreducible curves. For instance, (translation d'une )
%\end{rem}

%%%%%%%%%%%%%%%%%%%%%%%%%%%%%%%%%%%%%%%%%%%%%%%%%%%%%%%
%%%%%%%%%%%%%%%%%%%%%%%%%%%%%%%%%%%%%%%%%%%%%%%%%%%%%%%
%%%%%%%%%%%%%%%%%%%%%%%%%%%%%%%%%%%%%%%%%%%%%%%%%%%%%%%
%%%%%%%%%%%%%%%%%%%%%%%%%%%%%%%%%%%%%%%%%%%%%%%%%%%%%%%

\section{Codes from Jacobian surfaces}\label{CodesJacobian}

\subsection{Evaluation codes on algebraic surfaces}\label{DefEvCodes}
Let $X$ be a smooth, projective, absolutely irreducible algebraic surface defined over $\mathbb{F}_q$. Any  divisor $G$ on $X$ rational over $\mathbb{F}_q$ (i.e. $G$ is invariant under the action of $\mbox{Gal}(\overline{\mathbb{F}_q}/\mathbb{F}_q)$) such that the Riemann-Roch space 
$$L(G)=\{f\in \mathbb{F}_q(X)\setminus\{ 0\}\vert \mbox{Div}(f)+G\geq 0\}\cup\{ 0\}$$
is non-trivial defines a rational map to a projective space in the usual way: choose a basis $\mathcal{B}=\{ X_0,\dots ,X_{\ell -1}\}$ of the $\ell$-dimensional $\mathbb{F}_q$-vector space $L(G)$ and define 
\begin{eqnarray*}
\varphi_G : \quad X & \rightarrow & \mathbb{P}^{\ell -1} \\
P& \mapsto & (X_0(P):\dots :X_{\ell -1}(P))
\end{eqnarray*}
(a different choice of $\mathcal{B}$ would just change the projective space by an automorphism). By definition, $G$ is \emph{very ample} if $\varphi_G$ is an embedding. %, and conversely, any embedding of $X$ to a projective space (maximal???) is coming from a some divisor.
In what will follow, we assume that $G$ is very ample and see $X$ as an algebraic subvariety of $\mathbb{P}^{\ell -1}$, via $\varphi_G$. 

\bigskip

First, notice that the choice of the basis $\mathcal{B}$ induces an identification of $L(G)$ with the $\mathbb{F}_q$-vector space of linear forms in $\ell$ variables. For $f\in L(G)$, we denote by $\widetilde{f}$ the linear form such that $f=\widetilde{f}( X_0,\dots ,X_{\ell -1})$.

Now, let $X(\mathbb{F}_q)=\{ P_1,\dots ,P_n\}$ be an enumeration of the rational points on $X$ and for $i=1,\dots ,n$ and $j=0,\dots ,\ell -1$, fix some $P_{i,j}\in\mathbb{F}_q$ such that $P_i=(P_{i,0}:\dots : P_{i,\ell -1})$. This choice defines a linear map
\begin{eqnarray*}
ev: L(G) & \rightarrow & \mathbb{F}_q^n \\
f& \mapsto & (\widetilde{f}(P_{1,0},\dots ,P_{1,\ell -1}),\dots ,\widetilde{f}(P_{n,0},\dots ,P_{n,\ell -1})).
\end{eqnarray*}
The \emph{evaluation code} $C(X,G)$ is defined to be the image of the map $ev$. It is independent from the choice of $\mathcal{B}$ and of the coordinates $P_{i,j}$.

\bigskip

Given a function $f\in L(G)\setminus\{ 0\}$, we can consider the effective rational divisor 
$$D_f=\mbox{Div}(f)+G.$$
The divisor $D_f$ is the pullback of the hyperplane divisor of $\mathbb{P}^{\ell -1}$ defined by $\widetilde{f}$ (see \cite{sha}, III.1.4).  Therefore, the rational points on its support are exactly the $P_i$'s for which $\widetilde{f}(P_{i,0},\dots ,P_{i,\ell -1})=0$. 

As $D_f$ is rational over $\mathbb{F}_q$, the absolute Galois group $\mbox{Gal}(\overline{\mathbb{F}_q}/\mathbb{F}_q)$ acts on the set of its prime components, and letting $D_{1,f},\dots ,D_{k,f}$ be the sums of the elements in each of the $k$ orbits, we can write
\begin{eqnarray}\label{DecompDiv}
D_f=\sum_{i=1}^kn_iD_{i,f},
\end{eqnarray}
for some positive integers $n_1,\dots ,n_k$. The $D_{i,f}$'s define irreducible curves over $\mathbb{F}_q$ (which can be possibly reducible over $\overline{\mathbb{F}_q}$), again denoted by $D_{i,f}$. The discussion above shows that the number of zero coordinates of the codeword $ev(f)$ is at most
$$N(f)=\sum_{i=1}^k\# D_{i,f}(\mathbb{F}_q).$$
Therefore, the minimum distance of $C(X,G)$ satisfies
\begin{eqnarray}\label{MinDist}
d\geq \# X(\mathbb{F}_q)-\max_{f\in L(G)\setminus\{ 0\}}N(f).
\end{eqnarray}

%%%%%%%%%%%%%%%%%%%%%%%%%%%%%%%%%%%%%%%%%%%%%%%%%%%%%%%
%%%%%%%%%%%%%%%%%%%%%%%%%%%%%%%%%%%%%%%%%%%%%%%%%%%%%%%

\subsection{On the parameters of some codes on Jacobian surfaces}

Now, we focus our attention on the case where $X=J_C$ is the Jacobian of a smooth, projective, absolutely irreducible genus $2$ curve $C/\mathbb{F}_q$ with at least one rational point. Such a rational point defines an embedding $C\hookrightarrow J_C$, so we can consider $C$ as a divisor on $J_C$. The aim of this section is to derive from (\ref{MinDist}) a lower bound for the minimal distance of codes of the form $C(J_C,G)$, where $G$ is a very ample divisor which is numerically equivalent to a positive multiple $rC$ of $C$.  Notice that the divisor $rC$ is very ample for $r\geq 3$ (\cite{mum}, III.17).
%For instance, we can take $G$ to be a sum of translates of $C$.

\bigskip

If the right hand side of (\ref{MinDist}) is positive, then the map $ev$ must be injective and the dimension of $C(J_C,G)$ is equal to the dimension of the $\mathbb{F}_q$-vector space $L(G)$ defined in the former section. This last quantity can be computed using the Riemann-Roch Theorem for surfaces, which takes a particularly simple form in the case of very ample divisors on abelian surfaces (see \cite{mil2}, Th. 13.3 and \cite{mum}, III.16): 
$$\dim_{\mathbb{F}_q}L(G)  =  G^2/2.$$
Using (\ref{Adjunction}), we deduce that if the right hand side of (\ref{MinDist}) is positive, we have
\begin{eqnarray*}
\dim_{\mathbb{F}_q}C(J_C,G) & = & G^2/2\\
 & = & r^2C^2/2\\
 & = & r^2.
\end{eqnarray*}

\bigskip

We now give an estimate of the minimum distance of $C(J_C,G)$. In order to apply (\ref{MinDist}), we need some preliminary result about the components of an effective divisor linearly equivalent to $G$.

\begin{prop}\label{NbComp}
Let $D$ be an effective divisor rational over $\mathbb{F}_q$ and linearly equivalent to $G$, let
$$D=\sum_{i=1}^kn_iD_{i}$$
be its decomposition as a sum of orbits, as in (\ref{DecompDiv}) and let $\pi_i$ be the arithmetic genus of $D_i$, $i=1,\dots ,k$. 
\begin{enumerate}
\item We have
$$\sum_{i=1}^kn_i\sqrt{\pi_i-1}\leq r.$$

\item If $J_C$ is simple, then $\pi_i\geq 2$, $i=1,\dots ,k$ and therefore $k\leq r$.
\end{enumerate}
\end{prop}
\begin{proof}
We have
\begin{eqnarray*}
2r & = & rC.C\\
 & = & G.C\\
 & = & D.C\\
 & = & \sum_{i=1}^kn_iD_{i}.C.
\end{eqnarray*}
Since $C$ is ample, we have
$$(D_i^2)(C^2)\leq (D_i.C)^2,$$
for $i=1,\dots ,k$ (apply the Hodge Index Theorem to $(C^2)D_i-(D_i.C)C$, see \cite{hart}, V.1). Then, using (\ref{Adjunction}) and the fact that $\pi_i\geq 1$, $i=1,\dots ,k$ (see Section \ref{AboutWP}), we get
$$D_i.C\geq \sqrt{2\times 2-2}\sqrt{2\pi_i-2}=2\sqrt{\pi_i-1}.$$
This proves the first part of the proposition.

\bigskip

The second part of the proposition follows from the fact that there are no projective irreducible curve $E/\mathbb{F}_q$ of arithmetic genus $1$  lying on a simple abelian surface. If $E$ is absolutely irreducible, this statement is straightforward from Section \ref{AboutWP}. 

Else, let $E_1$ an absolutely irreducible component of $E$. In the light of the Adjunction Formula (\ref{Adjunction}), it is enough to prove that the self-intersection $E^2$ is positive. By Remark \ref{RqIntPos}, we have $E^2\geq E_1^2$, so the statement is also obvious if $E_1$ has arithmetic genus greater than $1$.

Assume that $E_1$ has arithmetic genus $1$. Let $E_2$ be the image of $E_1$ by the Frobenius. By hypothesis, $E_2$ is an absolutely irreducible component of $E$ distinct from $E_1$. We have $E^2\geq E_1.E_2$, so it is enough to prove that $E_1.E_2$ is nonzero. We will use the fact that a divisor on an abelian variety is algebraically equivalent (and therefore numerically equivalent) to its translates (see \cite{mil2}). Let ${E'}_1$ be a translate of $E_1$ passing through zero and ${E'}_2$ be the image of ${E'}_1$ by the Frobenius. Then ${E'}_2$ is clearly a translate of $E_2$, so $E_1.E_2={E'}_1.{E'}_2$. Moreover, ${E'}_1$ has arithmetic genus $1$, so it is not defined over $\mathbb{F}_q$ and thus, ${E'}_1$ is distinct from ${E'}_2$. Therefore we have ${E'}_1.{E'}_2\geq \# ( {E'}_1\cap {E'}_2)(\overline{\mathbb{F}_q})$, and this last intersection is nonempty, since it contains the zero element. This concludes the proof.
\end{proof}

\bigskip

\begin{theo}\label{MinDistJac}
With the notations above, if $J_C$ is simple, then the minimum distance of $C(J_C,G)$ satisfies
$$d\geq \# J_C(\mathbb{F}_q)- \max \left\{\# C(\mathbb{F}_q)+(r^2-1)[2\sqrt{q}],\quad r\# C(\mathbb{F}_q)\right\}.$$
In particular, if
$$r\leq \frac{\# C(\mathbb{F}_q)}{[2\sqrt{q}]}-1,$$
then we have
$$d\geq \# J_C(\mathbb{F}_q)-r\# C(\mathbb{F}_q).$$
\end{theo}
\begin{proof}
Given $f\in L(G)\setminus\{ 0\}$, let
$$D_f=\sum_{i=1}^kn_iD_{i,f}$$
be the decomposition of $D_f=\mbox{Div}(f)+G$, as in (\ref{DecompDiv}) and let $\pi_1,\dots ,\pi_k$ the respective arithmetic genera of $D_{1,f},\dots ,D_{k,f}$. We have $\pi_i\geq 2$, $i=1,\dots ,k$ (see Proposition \ref{NbComp}). Applying Proposition \ref{NbRatPts} to each $D_{i,f}$, we get
\begin{eqnarray}\label{MajNf}
N(f)= \sum_{i=1}^k\# D_{i,f}(\mathbb{F}_q)
 & \leq &  k(\# C(\mathbb{F}_q)-2[2\sqrt{q}])+[2\sqrt{q}]\sum_{i=1}^k\pi_i
\end{eqnarray}
We start by giving an upper bound for $\sum_{i=1}^k\pi_i$, depending on $k$. For $i=1,\dots ,k$, set $$s_i=\sqrt{\pi_i-1}-1.$$ 
The $s_i$'s are non-negative real numbers and thus
$$\sum_{i=1}^ks_i^2  \leq  \left( \sum_{i=1}^ks_i\right)^2.$$
Moreover, using Proposition \ref{NbComp}, we get
$$\sum_{i=1}^ks_i = \left(\sum_{i=1}^k\sqrt{\pi_i-1}\right) -k \leq r-k.$$
Therefore, since $\pi_i=(s_i+1)^2+1=s_i^2+2s_i+2$, $i=1,\dots ,k$, we have
\begin{eqnarray*}
\sum_{i=1}^k\pi_i & = & \sum_{i=1}^ks_i^2+2\sum_{i=1}^ks_i+2k\\
 & \leq & \left( \sum_{i=1}^ks_i\right)^2+2\sum_{i=1}^ks_i+2k\\
 & \leq & (r-k)^2+2(r-k)+2k\\
 & = & (r-k)^2+2r.
\end{eqnarray*}
Now, by combining this last inequality with (\ref{MajNf}), we get
\begin{eqnarray}\label{Maj2Nf}
N(f) & \leq &  k(\# C(\mathbb{F}_q)-2[2\sqrt{q}])+((r-k)^2+2r)[2\sqrt{q}].
\end{eqnarray}
The right hand side of (\ref{Maj2Nf}) defines a function $\phi$ of $k$ on the closed interval $[1;r]$, which is a polynomial of degree $2$ with positive leading coefficient, and therefore takes its maximum value at an extremity of its domain. In other words, we have
\begin{eqnarray*}
N(f) & \leq &  \max \{\phi (1),\phi (r)\}
\end{eqnarray*}
where
\begin{eqnarray*}
\phi (1) & = & \# C(\mathbb{F}_q)-2[2\sqrt{q}]+((r-1)^2+2r)[2\sqrt{q}]=\# C(\mathbb{F}_q)+(r^2-1)[2\sqrt{q}]\\
\phi (r) & = & r(\# C(\mathbb{F}_q)-2[2\sqrt{q}])+2r[2\sqrt{q}]=r\# C(\mathbb{F}_q).
\end{eqnarray*}
Then the result follows from (\ref{MinDist}).

To conclude, notice that we have
\begin{eqnarray*}
\phi (1)\leq \phi (r) & \Leftrightarrow & (r^2-1)[2\sqrt{q}]\leq (r-1)\# C(\mathbb{F}_q)\\
 & \Leftrightarrow & (r+1)[2\sqrt{q}]\leq \# C(\mathbb{F}_q)\\
 & \Leftrightarrow & r\leq \frac{\# C(\mathbb{F}_q)}{[2\sqrt{q}]}-1.
\end{eqnarray*}
\end{proof}

%%%%%%%%%%%%%%%%%%%%%%%%%%%%%%%%%%%%%%%%%%%%%%%%%%%%%%%
%%%%%%%%%%%%%%%%%%%%%%%%%%%%%%%%%%%%%%%%%%%%%%%%%%%%%%%

\vskip1cm

{\bf Acknowledgments:}
I would like to thank Peter Beelen for suggesting this work and for helpful discussions.
This work was supported by the Danish-Chinese Center for Applications of 
Algebraic Geometry in Coding Theory and Cryptography.

%%%%%%%%%%%%%%%%%%%%%%%%%%%%%%%%%%%%%%%%%%%%%%%%%%%%%%%
%%%%%%%%%%%%%%%%%%%%%%%%%%%%%%%%%%%%%%%%%%%%%%%%%%%%%%%

%%%%%%%%%%%%%%%%%%%%%%%%%%%%%%%%%%%%%%%%%%%%%%%%%%%%%%%
\end{document}